\documentclass[11pt]{article}
\usepackage[english]{babel}
\usepackage{amsmath,amsfonts,amsthm,amssymb,latexsym,wasysym}
\usepackage{color,graphicx,fancyhdr,mathrsfs,nicefrac,framed}
\usepackage{tikz}
\usepackage{enumerate,comment}
\usepackage{url}
\usepackage{subfigure}
\usepackage[colorlinks=true,citecolor=blue]{hyperref}
\usepackage{indentfirst}
\usepackage{multirow}
\usepackage{stackrel}

\usepackage[toc,page]{appendix}
\usepackage{natbib}
\usepackage[affil-it]{authblk}

\theoremstyle{plain}
\newtheorem{theorem}{Theorem}[section]
\newtheorem{lemma}[theorem]{Lemma}
\newtheorem{proposition}[theorem]{Proposition}

\theoremstyle{definition}
\newtheorem{definition}[theorem]{Definition}

\newtheorem{remark}[theorem]{Remark}

\newcommand{\esssup}{\mathop {\rm ess\,sup}\nolimits}
\newcommand{\essinf}{\mathop {\rm ess\,inf}\nolimits}

\newcommand{\VaR}{\mathop {\rm VaR}\nolimits}

\newcommand{\ES}{\mathop {\rm ES}\nolimits}

\newcommand{\LVaR}{\mathop {\rm LVaR}\nolimits}

\newcommand{\rag}[1][]{\rho_{g_{#1}}}

\newcommand{\precSSD}{\le_{\rm ssd}}

\newcommand{\cA}{\mathcal A}

\newcommand{\cF}{\mathcal F}
\newcommand{\cG}{\mathcal G}
\newcommand{\cH}{\mathcal H}

\newcommand{\cP}{\mathcal P}

\newcommand{\cU}{\mathcal U}

\newcommand{\cX}{\mathcal X}

\newcommand{\E}{\mathbb E}

\newcommand{\N}{\mathbb N}
\newcommand{\probp}{\mathbb P}
\newcommand{\probq}{\mathbb Q}
\newcommand{\R}{\mathbb R}

\newcommand{\Q}{\mathbb{Q}}

\newcommand{\X}{\mathcal{X}}

\newcommand{\p}{\mathbb{P}}

\renewcommand{\ge}{\geqslant}
\renewcommand{\le}{\leqslant}
\renewcommand{\geq}{\geqslant}
\renewcommand{\leq}{\leqslant}

\newcommand{\ic}{\Box}
\newcommand{\icn}[2]{\stackrel[\scriptscriptstyle {#1}=1]{\scriptscriptstyle {#2}}{\Box}}

\newcommand{\risk}{\ES}
\newcommand{\rg}{\ES^g}
\newcommand{\rh}{\ES^h}
\newcommand{\rng}[1]{\risk^{#1}} 

\newcommand{\SSD}{\ge_{\rm SSD}}

\def\d{\mathrm{d}}

\makeatletter
\let\@fnsymbol\@arabic
\makeatother

\usepackage[compact]{titlesec}

\makeatletter
\DeclareRobustCommand{\bsquare}{%
  \mathop{\vphantom{\sum}\mathpalette\bigstar@\relax}\slimits@
}

\newcommand{\bigstar@}[2]{%
  \vcenter{%
    \sbox\z@{$#1\sum$}%
    \hbox{\resizebox{.9\dimexpr\ht\z@+\dp\z@}{!}{$\m@th\dsquare$}}%
  }%
}
\makeatother

\newcommand{\dsquare}{\mathop{  \square} \displaylimits}

\begin{document}

 \begin{center}
{\Large Adjusted Expected Shortfall}
\\
Matteo Burzoni\\
{\em Department of Mathematics, University of Milan, Italy}\\
\url{matteo.burzoni@unimi.it}\\
Cosimo Munari\\
{\em Center for Finance and Insurance and Swiss Finance Institute, University of Zurich, Switzerland}\\
\url{cosimo.munari@bf.uzh.ch}\\
Ruodu Wang\footnote{Ruodu Wang is supported by Natural Sciences and Engineering Research Council of Canada (RGPIN-2018-03823, RGPAS-2018-522590).}\\
{\em Department of Statistics and Actuarial Science, University of Waterloo, Canada }\\
\url{wang@uwaterloo.ca}\\
\today
\end{center}
 \begin{abstract}
\noindent
We introduce and study the main properties of a class of convex risk measures that refine Expected Shortfall by simultaneously controlling the expected losses associated with different portions of the tail distribution. The corresponding adjusted Expected Shortfalls quantify risk as the minimum amount of capital that has to be raised and injected into a financial position $X$ to ensure that Expected Shortfall $\ES_p(X)$ does not exceed a pre-specified threshold $g(p)$ for every probability level $p\in[0,1]$. Through the choice of the benchmark risk profile $g$ one can tailor the risk assessment to the specific application of interest. We devote special attention to the study of risk profiles defined by the Expected Shortfall of a benchmark random loss, in which case our risk measures are intimately linked to second-order stochastic dominance.
\end{abstract}

\parindent 0em \noindent

\section{Introduction}

\noindent In this paper we introduce and discuss the main properties of a new class of quantile-based risk measures. Following the seminal paper by \cite{ADEH99}, we view a risk measure as a capital requirement rule. More precisely, we quantify risk as the minimal amount of capital that has to be raised and invested in a pre-specified financial instrument (which is typically taken to be risk free) to confine future losses within a pre-specified acceptable level of security. Value at Risk (VaR) and Expected Shortfall (ES) are the most prominent examples of monetary risk measures in the above sense. Throughout, we always adopt the convention to assign positive values to losses. Under VaR, a financial position is acceptable if its loss probability does not exceed a given threshold. In line with our convention, this means that VaR coincides the lower quantile of the underlying distribution at an appropriate level. Under ES, a financial position is acceptable if, on average, it does not produce a loss beyond a given VaR. In the banking regulatory sector, the Basel Committee has recently decided to move from VaR at level $99\%$ to ES at level $97.5\%$ for the measurement of financial market risk. In the insurance regulatory sector, VaR at level $99.5\%$ is the reference risk measure in the Solvency II and in the forthcoming Insurance Capital Standard framework while ES at level $99\%$ is the reference risk measure in the Swiss Solvency Test framework. In the past 20 years, an impressive body of research has investigated the relative merits and drawbacks of VaR and ES at both a theoretical and a practical level. This investigation led to a better understanding of the properties of these two risk measures at the same time triggering a variety of new research questions about risk measures in general. We refer to early work on ES in \cite{acerbi2002coherence} \cite{acerbi2002spectral}, \cite{frey2002var}, and \cite{rockafellar2002conditional} (where ES was called Conditional VaR). Some recent contributions to the broad investigation on whether and to what extent VaR and ES meet regulatory objectives are \cite{KochMunari2016}, \cite{EmbrechtsLiuWang2018}, \cite{weber2018solvency}, \cite{BBM18}, \cite{BaesKochMunari2020}, and \cite{WZ20}. For robustness problems concerning VaR and ES, see, e.g., \cite{CDS10} and \cite{KratschmerSchiedZahle2014}, and for their backtesting, see, e.g., \cite{Ziegel2016}, \cite{DuEscanciano2017}, and \cite{KratzLokMcNeil2018}.

\smallskip

A fundamental difference between VaR and ES is that, by definition, VaR is completely blind to the behavior of the loss tail beyond the reference quantile whereas ES depends on the whole tail beyond it. It is often argued that this difference, together with the convexity property, makes ES a superior risk measure compared to VaR. In fact, this is the main motivation that led the Basel Committee to shift from VaR to ES in their market risk framework; see \cite{BASEL35}. However, every risk measure captures risk in a specific manner and, as such, is bound to possess some limitations. This is also the case of ES. Indeed, 
being essentially an average beyond a given quantile, ES can only provide an aggregate estimation of risk which, by its very definition, {\em does not distinguish across different tail behaviors with the same mean}. While in specific situations a finer risk classification can be obtained by means of other risk measures, including spectral and deviation risk measures, our goal is to introduce a general class of convex risk measures that help make that distinction by using {\em ES as their fundamental building block}. The advantage of this approach is that it can be directly linked to a regulatory framework based on ES. To this end, we construct a risk measure that is sensitive to changes in the {\em ES profile} of a random variable $X$, i.e., the curve of ES
\[
p \mapsto \ES_p(X)
\]
viewed as a function of the underlying confidence level. More specifically, we  ``adjust'' $\ES$ into
\[
\ES^g(X) := \sup_{p\in[0,1]}\{\ES_p(X)-g(p)\}
\]
where $g:[0,1]\to(-\infty,\infty]$ is a given increasing function. The risk measure $\ES^g$ is called the {\em adjusted ES with risk profile $g$} and is a monetary risk measure in the sense of \cite{ADEH99}. Indeed, the quantity $\ES^g(X)$ can be interpreted as the minimal amount of cash that has to be raised and injected into $X$ in order to ensure the following target solvency condition:
\[
\ES^g(X)\leq0 \ \iff \ \mbox{$\ES_p(X)\leq g(p)$ for every $p\in[0,1]$}.
\]
In this sense, the function $g$ defines the threshold between acceptable and unacceptable ES profiles. Interestingly, $\ES^g$ is a convex risk measure but is not coherent unless it reduces to a standard ES.

\smallskip

The goal of this paper is to introduce the class of adjusted ES's and discuss their main theoretical properties. In Section \ref{sect: definition} we provide a formal definition and a useful representation of adjusted ES together with a number of illustrations. The focus of Section~\ref{sect: basic properties} is on some basic mathematical properties. A special interesting case is when the risk profile $g$ is given by the ES of a benchmark random variable. We focus on this situation in Section \ref{sect: SSD} and show that such special adjusted ES's are strongly linked with second-order stochastic dominance. More precisely, they coincide with the monetary risk measures for which acceptability is defined in terms of carrying less risk, in the sense of second-order stochastic dominance, than a given benchmark random variable. In Section \ref{sect: SSD optimization} we focus on a variety of optimization problems featuring risk functionals either in the objective function or in the optimization domain and study the existence of optimal solutions in the presence of this type of risk measures. In each case of interest we are able to establish explicit optimal solutions.


\section{Introducing adjusted ES}
\label{sect: definition}

Throughout the paper we fix an atomless probability space $(\Omega,\mathcal F,\p)$ and denote by $L^1$ the space of (equivalent classes with respect to $\p$-almost sure equality of) $\p$-integrable random variables. For any two random variables $X,Y\in L^1$ we write $X\sim Y$ whenever $X$ and $Y$ are identically distributed. We adopt the convention that positive values of $X\in L^1$ correspond to losses. In this setting, Value at Risk (VaR) and Expected Shortfall (ES) are respectively defined as
$$
\VaR_p(X):=
\begin{cases}
\inf \{x\in \R\mid \p(X\le x)\ge p\} & \mbox{if} \ p\in(0,1],\\
\essinf X & \mbox{if} \ p=0,
\end{cases}
$$
$$
\ES_p(X):=
\begin{cases}
\frac 1 {1-p} \int_p^1 \VaR_q(X)\d q & \mbox{if} \ p\in[0,1),\\
\esssup X & \mbox{if} \ p=1.
\end{cases}
$$
The quantities $\VaR_p(X)$ and $\ES_p(X)$ represent the minimal amount of cash that has to be raised and injected into $X$ in order to ensure the following target solvency condition (for $0<p<1$):
\[
\VaR_p(X)\leq0 \ \iff \ \p(X\leq0)\geq p,
\]
\[
\ES_p(X)\leq0 \ \iff \ \int_p^1 \VaR_q(X)\d q\leq0.
\]
The VaR solvency condition requires that the loss probability of $X$ is capped by $1-p$ whereas the ES solvency condition states that there is no loss on average beyond the (left) $p$-quantile of $X$. 

\smallskip

The focus of the paper is on the following class of risk measures. Here and in the sequel, we denote by $\cG$ the set of all functions $g:[0,1]\to(-\infty,\infty]$ that are increasing (in the non-strict sense) and not identically $\infty$. Moreover, we use the convention $\infty-\infty=-\infty$.

\begin{definition}
Consider a function $g\in\cG$ and define the set
$$
\cA_g := \{X\in L^1 \mid \forall p\in[0,1], \ \ES_p(X)\leq g(p)\}.
$$
The functional $\rg:L^1\to(-\infty,\infty]$ defined by
$$
\rg(X) := \inf\{m\in\R \mid X-m\in\cA_g\}.
$$
is called the {\em $g$-adjusted Expected Shortfall (g-adjusted ES)}.
\end{definition}

\smallskip

To best appreciate the financial interpretation of the above risk measure, it is useful to consider the {\em ES profile} associated with a random variable $X\in L^1$, i.e., the function
$$
p \mapsto \ES_p(X).
$$
From this perspective, the function $g$ in the preceding definition can be interpreted as a threshold between acceptable (safe) and unacceptable (risky) ES profiles. In this sense, the set $\cA_g$ consists of all the positions with acceptable ES profile and the quantity $\rg(X)$ represents the minimal amount of capital that has to be injected into $X$ in order to align its ES profile with the chosen acceptability profile. For this reason, we will sometimes refer to $g$ as the target {\em ES profile} or, more generally, the target {\em risk profile}. If, for given $p\in[0,1]$, we consider the target ES profile
$$
g(q)=
\begin{cases}
0 & \mbox{if} \ q\in[0,p],\\
\infty & \mbox{if} \ q\in(p,1],
\end{cases}
$$
then $\rg(X)=\ES_p(X)$ for every random variable $X\in L^1$. In words, the standard ES is a special case of an adjusted ES. The next proposition highlights an equivalent but operationally preferable formulation of adjusted ES's which also justifies the chosen terminology.

\begin{proposition}
\label{prop: cash additive representation}
For every risk profile $g\in\cG$ and for every $X\in L^1$ we have
$$
\rg(X) = \sup_{p\in[0,1]}\{\ES_p(X)-g(p)\}.
$$
\end{proposition}
\begin{proof}
Fix $X\in L^1$ and note that for every $m\in\R$ the condition $X-m\in\cA_g$ is equivalent to
\[
\ES_p(X)-m = \ES_p(X-m) \leq g(p)
\]
for every $p\in[0,1]$. For $p=1$ both sides could be equal to $\infty$. However, in view of our convention $\infty-\infty=-\infty$, the above inequality holds if and only if $m\geq\ES_p(X)-g(p)$ for every $p\in[0,1]$. The desired representation easily follows.
\end{proof}

\smallskip

\begin{remark}
(i) In line with our main motivation, the adjusted ES is a tool that allows us to distinguish risks with the same tail expectation without leaving the world of ES. In the context of the discussion on tail risk triggered by \cite{BASEL35}, the authors of \cite{LiuWang2019} proposed the following way to quantify the degree of tail blindness of a risk measure: For a given $p\in(0,1)$, a functional $\rho:L^1\to(-\infty,\infty]$ satisfies the {\em $p$-tail property} if for all $X,Y\in L^1$
\[
\mbox{$\VaR_q(X)=\VaR_q(Y)$ for every $q\in[p,1)$} \ \implies \ \rho(X)=\rho(Y).
\]
In this case, $\rho$ does not distinguish between two random losses having the same (left) quantiles beyond level $p$. It is not difficult to prove that $\ES^g$ satisfies the $p$-tail property if and only if $g$ is constant on the interval $(0,p)$. This provides a simple way to tailor the tail sensitivity of $\ES^g$.

\smallskip

(ii) The definition of $\ES^g$ is reminiscent of the {\em Loss Value at Risk} (LVaR) introduced in \cite{BBM18}. In that case, one takes an increasing and right-continuous function $\alpha:[0,\infty)\to[0,1]$ (the so-called benchmark loss distribution) and defines the acceptance set by
\[
\cA_\alpha := \{X\in L^1 \mid \probp(X>x)\leq\alpha(x), \ \forall x\geq0\}.
\]
The corresponding LVaR is given by
\[
\LVaR_\alpha(X) := \inf\{m\in\R \mid X-m\in\cA_\alpha\}.
\]
The quantity $\LVaR_\alpha(X)$ represents the minimal amount of capital that has to be injected into the position $X$ in order to ensure that, for each loss level $x$, the probability of exceeding a loss of size $x$ is controlled by $\alpha(x)$. According to Proposition 3.6 in the cited paper, we can equivalently write
\begin{equation}\label{eq:LVaReq}
\LVaR_\alpha(X) = \sup_{p\in[0,1]}\left\{\VaR_{p}(X)-\alpha^{-1}_+(p)\right\},
\end{equation}
where $\alpha^{-1}_+$ is the right inverse of $\alpha$. This highlights the similarity with adjusted ES's.
\end{remark}

\smallskip

\begin{figure}[t]
\centering
\subfigure{
\includegraphics[width=0.48\textwidth]{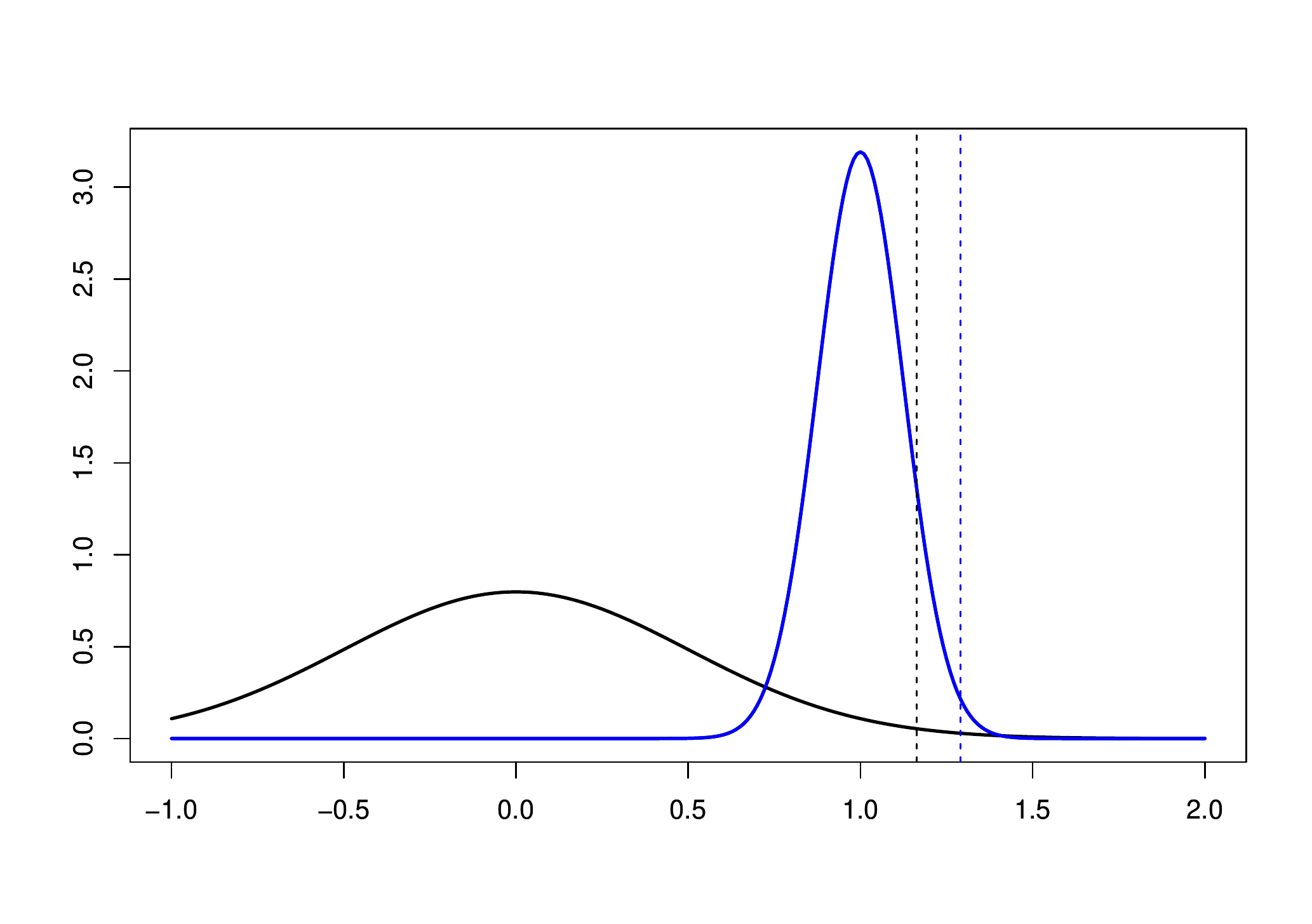}
}
\subfigure{
\includegraphics[width=0.48\textwidth]{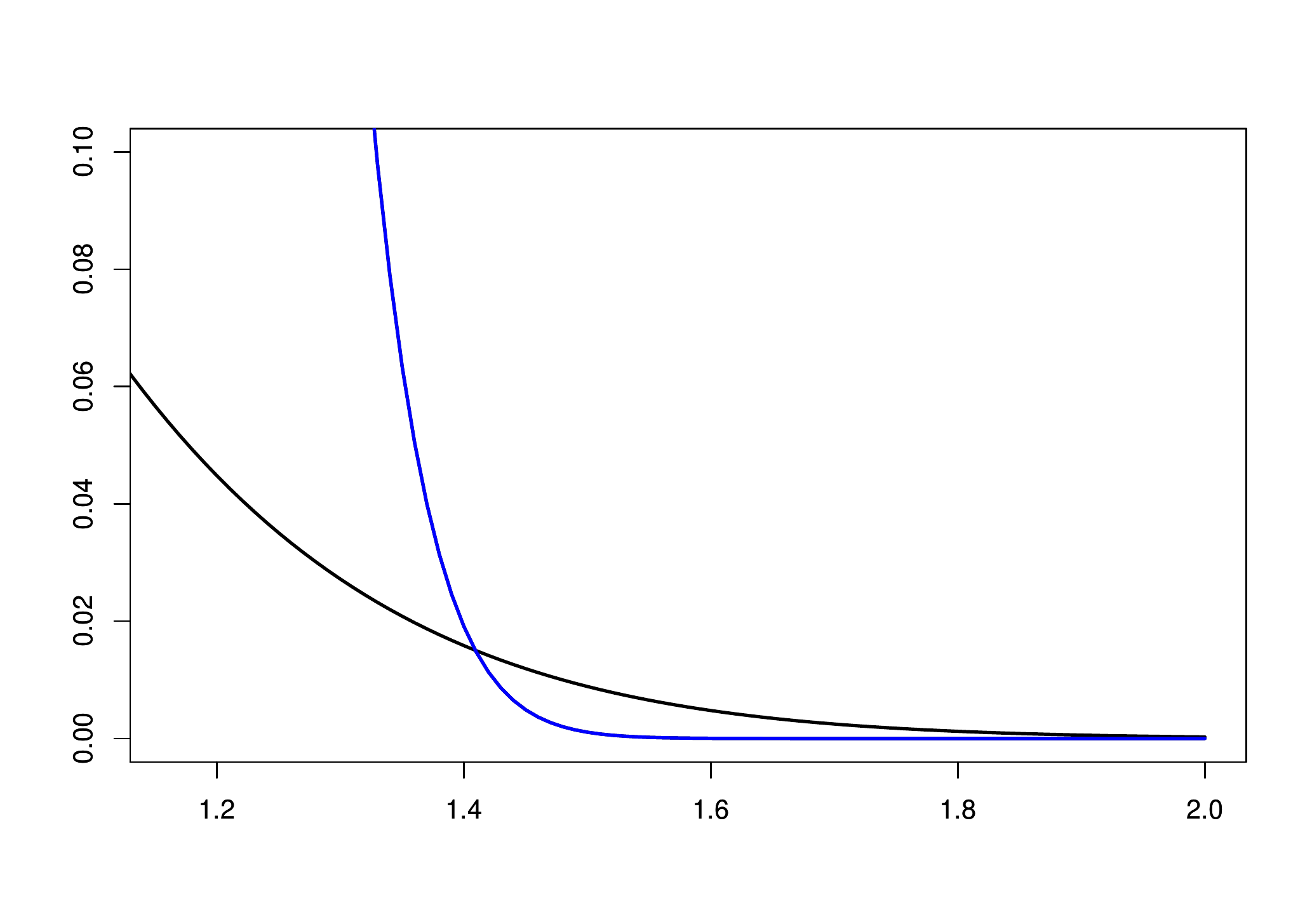}
}
\includegraphics[width=0.48\textwidth]{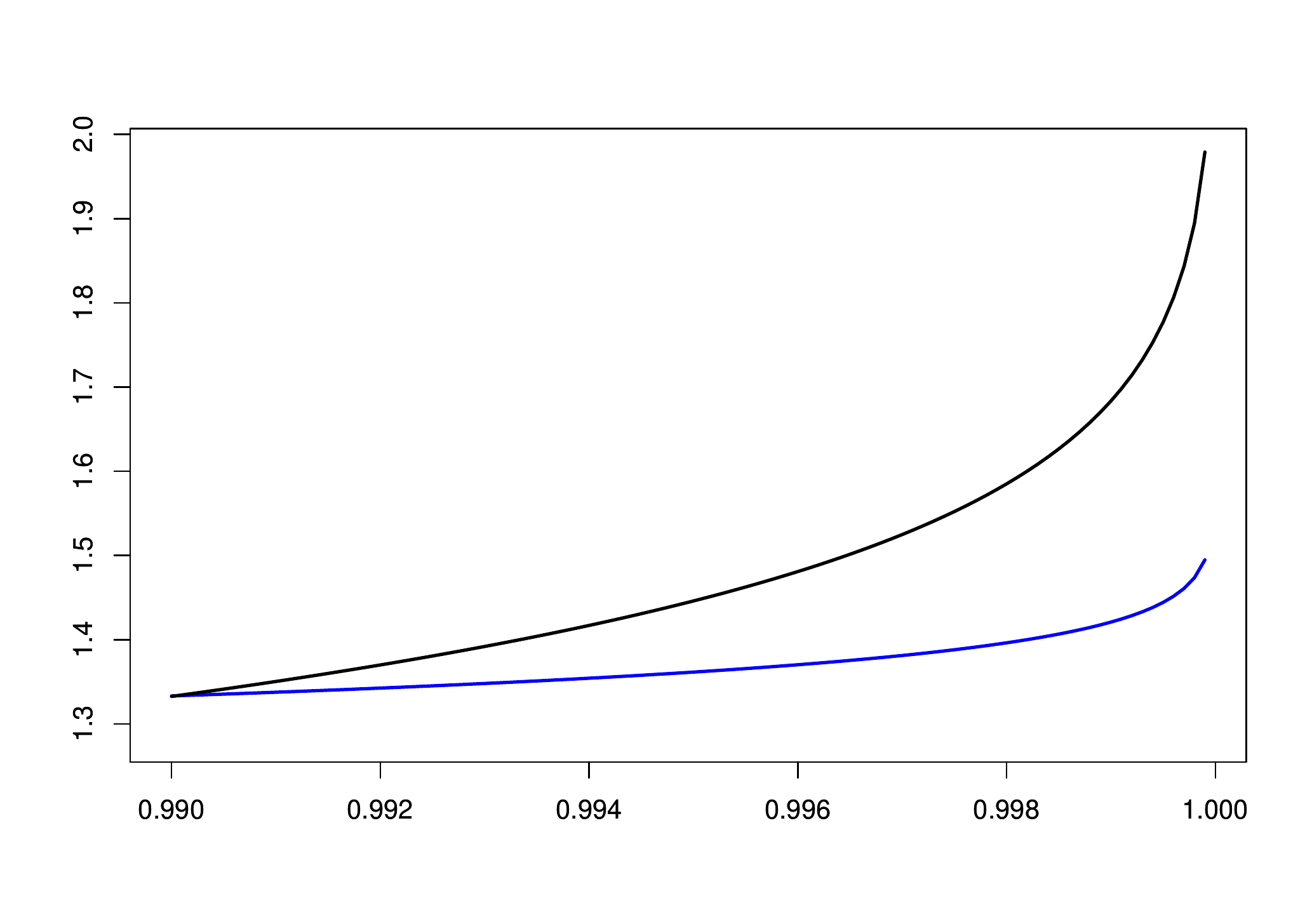}
\caption{Left: Density function of $X_1$ (blue) and $X_2$ (black). The vertical lines correspond to the respective $99\%$ quantiles. Right: Tails of of $X_1$ (blue) and $X_2$ (black) beyond the $99\%$ quantile. Below: $\ES$ profile of $X_1$ (blue) and $X_2$ (black) for $p\ge 0.99$.}
\label{img:distr}
\end{figure}

To illustrate the functioning of the adjusted ES, we consider the following simple example. Consider two normally distributed random variables $X_i\sim N(\mu_i,\sigma_i^2)$, with $\mu_1=1$, $\mu_2=0$, $\sigma_1=0.125$, $\sigma_2=0.5$. For every probability level $p\in(0,1)$ we have
\[
\ES_{p}(X_i)=\mu_i+\sigma_i\frac{\phi(\Phi^{-1}(p))}{1-p},
\]
where $\phi$ and $\Phi$ are, respectively, the density and the distribution function of a standard normal random variable. For $p=99\%$ the $\ES$ of both random variables is approximately equal to $1.33$. In Figure \ref{img:distr} we plot the two distribution functions. Despite having the same ES, the two risks are quite different mainly because of their different variance: The potential losses of $X_1$ tend to accumulate around its mean whereas those of $X_2$ are more disperse and can be significantly higher (compare the tails in Figure \ref{img:distr}). A closer look at the $\ES$ profile of both random variables shows that the ES profile of $X_1$ is more stable than that of $X_2$ (see again Figure \ref{img:distr}). A simple way to distinguish $X_1$ and $X_2$ while, at the same time, focusing on average losses beyond the $99\%$ quantile is to consider the adjusted ES with risk profile
\[
g(p)=
\begin{cases}
0 & \mbox{if} \ p\in[0,0.99],\\
0.1 & \mbox{if} \ p\in(0.99,0.9975],\\
\infty & \mbox{if} \ p\in(0.9975,1].\\
\end{cases}
\]
In this case, we easily obtain
\begin{equation}\label{eq:introES}
\ES^g(X_i) = \max\{\ES_{0.99}(X_i),\ES_{0.9975}(X_i)-0.1\} =
\begin{cases}
\ES_{0.99}(X_1) \approx 1.33 & \mbox{for} \ i=1,\\
\ES_{0.9975}(X_2)-0.1 \approx 1.45 & \mbox{for} \ i=2.
\end{cases}
\end{equation}
The focus of $\ES^g$ is still on the tail beyond the $99\%$ quantile. However, the risk measure $\ES^g$ is able to detect the heavier tail of $X_2$ and penalize it with a higher capital requirement. This is because $\ES^g$ is additionally sensitive to the tail beyond the $99.75\%$ quantile and penalizes any risk whose average loss on this far region of the tail is too large.

\smallskip

We use a similar target risk profile to compare the behavior of the classical ES and the adjusted ES on real data. We collect the S\&P 500 and the NASDAQ Composite indices daily log-returns (using closing prices) from January 01, 1999 to June 30, 2020. Each index has 5406 data points (publicly available from Yahoo Finance). We estimate the risk measures using a standard AR(1)-GARCH(1,1) model with t innovations (see Chapter 4 of \cite{MFE15} for details). In line with Basel III guidelines, to obtain less volatile outcomes we compute average risk measure estimates based on a 60-days moving window. We consider the risk profile function

\begin{figure}[t]
\centering
\subfigure{
\includegraphics[width=0.48\textwidth]{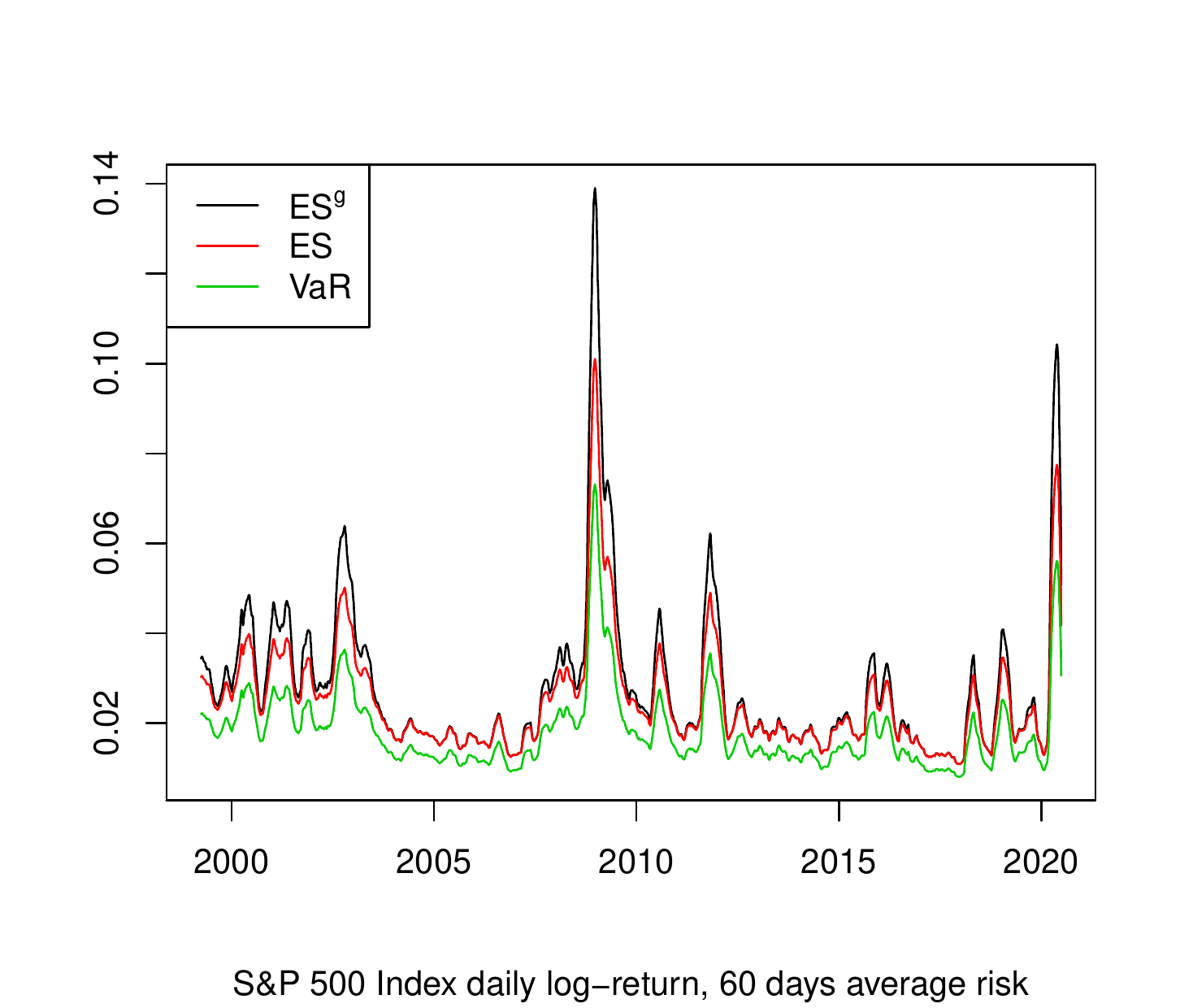}
}
\subfigure{
\includegraphics[width=0.48\textwidth]{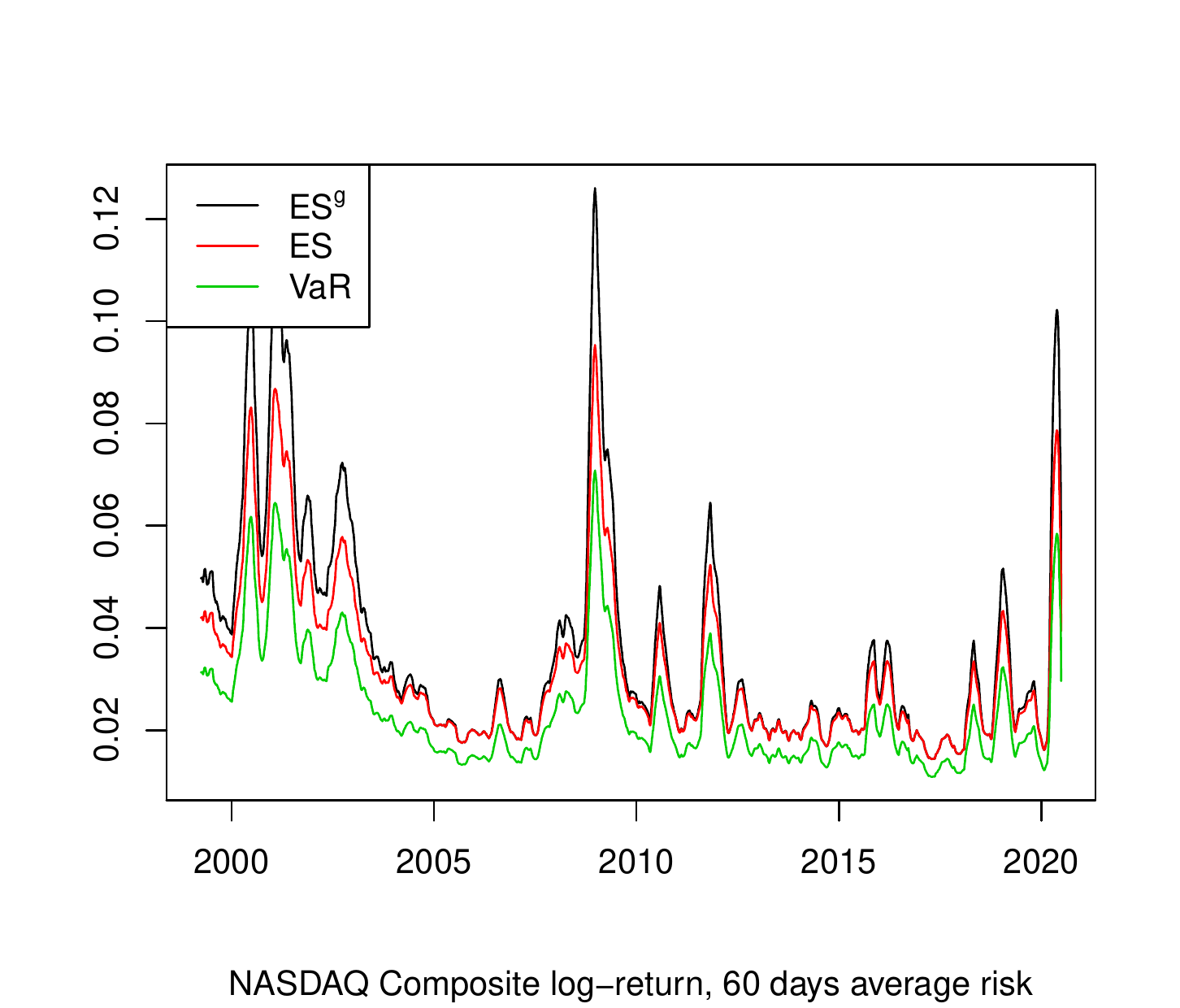}
}
\caption{Estimated $\ES_{0.95}$, $\ES^g$, and $\VaR_{0.95}$ for S\&P 500 and NASDAQ.}
\label{fig:empirical}
\end{figure}

\[
g(p)=
\begin{cases}
0 & \mbox{if} \ p\in[0,0.95],\\
0.01 & \mbox{if} \ p\in(0.95,0.99],\\
\infty & \mbox{if} \ p\in(0.99,1],\\
\end{cases}
\]
which yields
$$
\ES^g(X) = \max \{\ES_{0.95}(X), \ES_{0.99}(X)-0.01\}
$$
similar to \eqref{eq:introES} in a different context. The numbers $0.95$, $0.99$, and $0.01$ that appear in $g$ are chosen for the ease of illustration only. The 20-year estimated values of $\ES$ at level $95\%$ and $\ES^g$, as well as those of $\VaR$ at level $95\%$, are plotted in Figure \ref{fig:empirical}. As we can see from the numerical results on both S\&P 500 and NASDAQ, the estimated values of $\ES^g$ and the reference $\ES$ approximately agree with each other during most of the considered time horizon. However, during periods of significant financial stress, such as the dot-com bubble in 2000, the subprime crisis in 2008, and the COVID-19 crisis in early 2020, $\ES^g$ is visibly larger than the reference $\ES$. This illustrates that $\ES^g$ may capture tail risk in a more appropriate way than $\ES$, especially under financial stress.


\subsubsection*{Choosing the target ES profile}

As illustrated above, a key feature of adjusted ES is the flexibility in the choice of the target risk profile $g$. Indeed, the same random loss can be considered more or less relevant depending on a variety of factors, including the availability of hedging strategies or other risk mitigation tools in the underlying business sector. The choice of $g$ can be therefore tailored to the particular area of application by assigning different weights to different portions of the reference tail. Two examples are especially relevant. On the one hand, we consider a continuous risk profile of the form
\[
g(p) = \ES_p(L),
\]
where $L$ is a benchmark random loss. In this case, we have
\[
\ES^g(X) = \sup_{p\in[0,1]}\{\ES_p(X)-\ES_p(L)\}.
\]
The associated target solvency condition reads:
\[
\ES^g(X)\leq0 \ \iff \ \mbox{$\ES_{p}(X)\leq\ES_p(L)$ for every $p\in[0,1]$}.
\]
This choice of $g$ seems appropriate in the context of portfolio risk management. The distribution of the random loss $L$ may belong to a class of benchmark distributions and the adjusted ES corresponds to the smallest amount of cash that has to be raised and injected in the portfolio to shift its profit and loss distribution until the new distribution dominates the benchmark distribution in the sense of second order stochastic dominance. In other words, the above adjusted ES incorporates second-order stochastic dominance into a monetary risk measure by
\[
\ES^g(X) = \inf\{m\in\R \mid X-m \SSD L\}
\]
where $\SSD$ denotes second-order stochastic domination. Despite the importance of such a concept, we are not aware of earlier attempts to explicitly construct monetary risk measures whose underlying acceptability condition is based on second-order stochastic dominance. This paper offers first results in this direction thereby preparing the theoretical ground for new contributions to the rich literature on the application of stochastic dominance to portfolio risk management, for which we refer to the survey by \cite{LevySSD} and to the more recent contributions by, e.g., \cite{ogryczak2002dual}, \cite{de2005reward}, and \cite{hodder2015improved}.

\smallskip

In the second example, we consider a piecewise constant function of the form
\begin{equation}
\label{eq: piecewise constant g}
g(p)=
\begin{cases}
r_1 & \mbox{if} \ p\in[0,p_1],\\
r_2 & \mbox{if} \ p\in(p_1,p_2],\\
\dots \\
r_n & \mbox{if} \ p\in(p_{n-1},p_n],\\
\infty & \mbox{if} \ p\in(p_n,1],
\end{cases}
\end{equation}
where $0=r_1<\cdots<r_{n-1}<\infty$ and $0<p_1<\cdots<p_n<1$. In this case, we have
\[
\ES^g(X) = \max_{i=1,\dots,n}\{\ES_{p_i}(X)-r_i\}.
\]
The associated target solvency condition reads:
\[
\ES^g(X)\leq0 \ \iff \ \mbox{$\ES_{p_i}(X)\leq r_i$ for every $i=1,\dots,n$}.
\]
The coefficients $r_1,\dots,r_n$ represent benchmark risk thresholds whereas $p_1,\dots,p_n$ correspond to some pre-specified confidence levels. Note that, by design, we always have
\[
\ES^g(X) \geq \ES_{p_1}(X).
\]
This choice of $g$ seems appropriate in the context of solvency regulation. If $p_1$ coincides with a reference regulatory level, e.g.\ $97.5\%$ in Basel III and $99\%$ in the Swiss Solvency Test, the adjusted ES is by design as stringent as the regulatory ES and the additional thresholds $r_2,\dots,r_n$ impose extra limitations to the amount of risk that a firm is allowed to take. In particular, different bounds can be imposed for the, e.g., one in a hundred times event, one in a thousand times event, and the one in a hundred thousand times event. These bounds may correspond to suitable fractions of available capital so that, in case of such adverse events, one can directly quantify the necessary cost for covering the underlying losses. In this way, the actual risk bounds would be firm specific but the rule to determine them would be the same for every company. This is reminiscent of the proposal about Loss VaR in \cite{BBM18}, with ES replacing VaR. It is worth pointing out that imposing additional constraints for higher risks may lead to lower the base regulatory requirement by taking $p_1$ strictly smaller than the reference regulatory level. By doing so, regulators may avoid penalizing firms that are particularly careful about their tail behavior.

\smallskip

A piecewise constant risk profile may be adopted also in other applications. We provide a simple illustration in the context of cyber risk. Differently from other operational risks, cyber risk has a strong geographical component. The empirical study \cite{cyber_insurability}, which takes into account 22,075 incidents reported between March 1971 and September 2009, reveals that ``Northern America has some of the lowest mean cyber risk and non-cyber risk losses, whereas Europe and Asia have much higher average losses despite Northern American companies experience more than twice as many (51.9 per cent) cyber risk incidents than European ﬁrms (23.2 per cent) and even more than twice as many as firms located on other continents''. A possible reason is that North American companies may be better equipped to protect themselves against such events. Cyber risk cannot be properly managed by a simple frequency-severity analysis. In the qualitative analysis of \cite{cyberbook}, many additional factors are identified including ease of discovery, ease of exploit, awareness and intrusion detection. The answers may very well depend on the specific sector if not on the specific firms under consideration. The choice of different reference risk profiles $g$ across companies might be a way to apply the theory of risk measures in the spirit of \cite{ADEH99} to the rather complex analysis of this type of risk. For example, it would be possible to set
\[
g(p)=
\begin{cases}
\ES_{0.99}(Z_1) & \mbox{if} \ p\in[0,0.99],\\
\ES_{0.999}(Z_2) & \mbox{if} \ p\in(0.99,0.999],\\
\ES_{0.9999}(Z_3) & \mbox{if} \ p\in(0.999,0.9999],\\
\infty & \mbox{otherwise},
\end{cases}
\]
where $Z_1,Z_2,Z_3$ are suitable benchmark random losses. The resulting adjusted ES is
\[
\ES^g(X)=\max\{\ES_{0.99}(X)-\ES_{0.99}(Z_1),\ES_{0.999}(X)-\ES_{0.999}(Z_2),\ES_{0.999}(X)-\ES_{0.999}(Z_3)\}.
\]
The associated target solvency condition is given by
\[
\ES^g(X)\leq0 \ \iff \
\begin{cases}
\ES_{0.99}(X)\leq\ES_{0.99}(Z_1),\\
\ES_{0.999}(X)\leq\ES_{0.999}(Z_2),\\
\ES_{0.9999}(X)\leq\ES_{0.9999}(Z_3).
\end{cases}
\]
The choice of $g$ should be motivated by specific cyber risk events (see \cite{cyberbook} for a categorization of likelihood/severity for different cyber attacks): The one in a hundred times event could be the malfunctioning of the server, the one in a thousand times event the stealing of the profile data of the clients, the one in a hundred thousand times event the stealing of the credit cards details of the customers. Note that it is possible to choose a single benchmark random loss or a different benchmark random loss for each considered incident. This choice could also be company specific so as to reflect the company's ability to react to the different types of cyber attacks. This is in line with \cite{cyber_insurability}, which says that ``Regarding size (of the average loss per event), we observe a U-shaped relation, that is, smaller and larger firms have higher costs than medium-sized. Possibly, smaller firms are less aware of and less able to deal with cyber risk, while large firms may suffer from complexity''.

\smallskip

While in principle a different risk category may call for a different choice of the acceptable ES profile $g$, it is sometimes important in practice to ensure a certain degree of comparability across risk assessments.\footnote{We thank an anonymous referee for stressing this important point.} Suppose for example that a bank wants to compare the exposure to different risks $X_1,\ldots,X_k$ arising from different business lines. In principle, each business unit may use a specific ES profile $g_j$. However, if the bank requires that $g_1=\cdots=g_k=0$ on $[0,p)$ for a common $p\in(0,1)$, we can write
\[
\ES^{g_j}(X_j) = \ES_p(X_j)+\underbrace{\ES^{g_j}(X_j)-\ES_p(X_j)}_{\geq0}.
\]
For each $X_j$, the first component in the decomposition is an $\ES$ with common confidence level $p$, which can be used for comparison. The exceedance term $\ES^{g_j}(X_j)-\ES_p(X_j)$ represents the extra amount of capital that is needed to cover the specific risk type. The above decomposition takes a more explicit form if each $g_j$ is a piecewise constant function as in \eqref{eq: piecewise constant g} with customized parameters $r^j_i$'s and $p^j_i$'s. If we take $p^1_1=\cdots=p^k_1=p$, then we obtain
\[
\ES^{g_j}(X_j) = \ES_p(X_j)+\max\bigg\{\max_{i=2,\dots,n}\{\ES_{p^j_i}(X_j)-\ES_p(X_j)-r^j_i\},0\bigg\}.
\]
In this case, the risk-specific component is activated only when $\ES_{p^j_i}(X_j)$ is larger than the penalized benchmark ES term $\ES_p(X_j)+r^j_i$ for some index $i$. The parameters $r^j_i$'s and $p^j_i$'s can be tailored, e.g., to the size of the underlying tails. This example can be easily adapted to include a different number of thresholds for each risk class, i.e., $n$ may also depend on $j$. The choice may depend, e.g., on the size of the available observation sample and the frequency of tail observations.


\section{Basic properties of adjusted ES}
\label{sect: basic properties}

In this section we discuss a selection of relevant properties of adjusted ES. It is a direct consequence of our definition that every adjusted ES is a monetary risk measure in the sense of \cite{FS16}, i.e., is monotone and cash additive. The other properties listed below are automatically inherited from the corresponding properties of ES. For every risk profile $g\in\cG$ the risk measure $\rg$ satisfies the following properties:
\begin{itemize}
\item {\em monotonicity}: $\rg(X)\leq\rg(Y)$ for all $X,Y\in L^1$ such that $X\leq Y$.
\item {\em cash additivity}: $\rg(X+m)=\rg(X)+m$ for all $X\in L^1$ and $m\in\R$.
\item {\em convexity}: $\rg(\lambda X+(1-\lambda)Y) \leq\lambda\rg(X)+(1-\lambda)\rg(Y)$ for all $X,Y\in L^1$ and $\lambda\in[0,1]$.
\item {\em law invariance}: $\rg(X)=\rg(Y)$ for all $X,Y\in L^1$ such that $X\sim Y$.
\item {\em normalization}: $\rg(0)=0$ if and only if $g(0)=0$.
\end{itemize}

\smallskip

Being convex and law invariant, every adjusted ES is automatically consistent with second-order stochastic dominance; see, e.g., \cite{bellini2020law}. In fact, the link between adjusted ES's and stochastic dominance is far stronger. Recall that for any random variables $X,Y\in L^1$ we say that {\em $X$ dominates $Y$ with respect to second-order stochastic dominance}, written $X\SSD Y$, whenever the following condition holds:
\[
\mbox{$\E[u(-X)]\ge\E[u(-Y)]$ for every increasing and concave function $u:\R\to\R$}.
\]
In the language of utility theory, this means that $X$ is preferred to $Y$ by every risk-averse agent (recall that positive values of a random variable represent losses). We refer to \cite{bookSSD} for a classical reference on stochastic dominance. By convexity and law invariance, for every risk profile $g\in\cG$ the risk measure $\ES^g$ satisfies:
\begin{itemize}
    \item {\em consistency with $\SSD$}: $\rg(X)\leq\rg(Y)$ for all $X,Y\in L^1$ such that $X\SSD Y$.
\end{itemize}
This implies that $\rg$ belongs to the class of \emph{consistent risk measures} as defined in \cite{MW16}. In fact, it is shown in that paper that {\em any} consistent risk measure can be expressed as an infimum of a collection of risk measures which, using the terminology of this paper, are precisely of adjusted ES type.

\begin{proposition}[Theorem 3.1 in \cite{MW16}]
Let $\rho:L^1\to(-\infty,\infty]$ be cash additive and consistent with $\SSD$. Then, there exists $\cH\subset\cG$ such that for every $X\in L^1$ we have
\[
\rho(X) = \inf_{g\in\cH}\ES^g(X).
\]
\end{proposition}

\smallskip

The above proposition shows that adjusted ES can be seen as the {\em building block for risk measures that are consistent with second-order stochastic dominance}. This class is large and includes, e.g., all law-invariant convex risk measures.

\smallskip

It is well known that, in addition to convexity, ES satisfies positive homogeneity. This qualifies it as a coherent risk measure in the sense of \cite{ADEH99}. In the next proposition we show that $\rg$ satisfies positive homogeneity only in the case where it coincides with some ES. In other words, with the exception of ES, the class of adjusted ES's consists of monetary risk measures that are convex but not coherent. 

\begin{proposition}\label{prop:coherent}
For every risk profile $g\in\cG$ the following statements are equivalent:
\begin{enumerate}[(a)]
  \item $\rg$ is positively homogeneous, i.e., $\rg(\lambda X)=\lambda\rg(X)$ for all $X\in L^1$ and $\lambda\in(0,\infty)$.
  \item $g(0)=0$ and $g(p)\in(0,\infty)$ for at most one $p\in(0,1]$.
  \item $\rg=\ES_p$ where $p=\sup\{q\in[0,1] \mid g(q)=0\}$.
\end{enumerate}
\end{proposition}
\begin{proof}
``(a)$\Rightarrow$(b)": Since $\rg$ is positively homogeneous we have
\[
\lambda g(0) = -\lambda\rg(0) = -\rg(\lambda0) = -\rg(0) = g(0)
\]
for every $\lambda\in(0,\infty)$. As $g(0)<\infty$ by our assumptions on the class $\cG$, we must have $g(0)=0$. Now, assume by way of contradiction that $0<g(p_1)\leq g(p_2)<\infty$ for some $0<p_1<p_2\leq1$. Take now $q\in(p_1,p_2)$ and $b\in(0,g(p_1))$ and set
\[
a = \min\left\{-\frac{(1-q)b}{p-p_1},\inf_{p\in[0,p_1)}\frac{(1-p)g(p)-b(1-q)}{q-p}\right\}.
\]
Note that $a<0$. Since the underlying probability space is assumed to be atomless, we can always find a random variable $X\in L^1$ satisfying
\[
F_X(x)=
\begin{cases}
0 & \mbox{if $x\in(-\infty,a)$},\\
q & \mbox{if $x\in[a,b)$},\\
1 & \mbox{if $x\in[b,\infty)$}.
\end{cases}
\]
Note that, for every $p\in[0,p_1)$, the definition of $a$ implies
\[
\frac{(1-p)g(p)-b(1-q)}{q-p} \geq a.
\]
Moreover, for every $p\in[p_1,q)$, the choice of $b$ implies
\[
\frac{(1-p)g(p)-b(1-q)}{q-p} \geq \frac{(1-p)g(p_1)-b(1-q)}{q-p} \geq \frac{(1-p)b-b(1-q)}{q-p} = b \geq a.
\]
As a result, for every $p\in[0,q)$ we obtain
\[
\ES_p(X) = \frac{a(q-p)+b(1-q)}{1-p} \leq g(p).
\]
Similarly, for every $p\in[q,1]$ we easily see that
\[
\ES_p(X) = b < g(p_1) \leq g(q) \leq g(p).
\]
This yields $\rg(X)\leq0$. However, taking $\lambda>0$ large enough delivers
\[
\rg(\lambda X) = \sup_{p\in[0,1]}\{\lambda\ES_p(X)-g(p)\} \geq \lambda\ES_q(X)-g(q) = \lambda b-g(q) > 0
\]
in contrast to positive homogeneity. As a consequence, we must have $p_1=p_2$ and thus (b) holds.

\smallskip

``(b)$\Rightarrow$(c)":  Set $q=\sup\{p\in[0,1] \mid g(p)=0\}$. Note that $q\in[0,1]$. Clearly, we have $g(p)=0$ for every $p\in[0,q)$ and $g(p)=\infty$ for every $p\in(q,1]$ by assumption. Take an arbitrary $X\in L^1$. From the definition of $\ES$ and the continuity of the integral, it follows that $p\mapsto\ES_p(X)$ is continuous. As a result, we obtain
$$
\rg(X) = \sup_{p\in[0,q]}\{\ES_p(X)-g(p)\} = \sup_{p\in[0,q]}\ES_p(X) = \ES_q(X).
$$

\smallskip

``(c)$\Rightarrow$(a)": The implication is clear.
\end{proof}

\smallskip

An adjusted ES is convex but, unless it coincides with a standard ES, not subadditive. It is therefore natural to focus on infimal convolutions of adjusted ES's, which are important tools in the study of optimal risk sharing and capital allocation problems involving non-subadditive risk measures; see, e.g., \cite{BarrieuElKaroui2005}, \cite{BurgertRueschendorf2008}, \cite{FilipovicSvindland2008} for results in the convex world and \cite{EmbrechtsLiuWang2018} for results beyond convexity.

\begin{definition}\label{def:infconv}
Let $n\in\N$ and consider $\rho_1,\dots,\rho_n:L^1\to(-\infty,\infty]$. For every $X\in L^1$ we set
\[
\mathcal{S}^n(X) := \left\{(X_1,\ldots,X_n)\in L^1\times\cdots\times L^1 \,\bigg\vert\, \sum_{i=1}^n X_i=X\right\}.
\]
The map $\icn{i}{n}\rho_i:L^1\to[-\infty,\infty]$ defined by
\[
\icn{i}{n}\rho_i(X):=\inf\left\{\sum_{i=1}^n\rho_i(X_i) \,\bigg\vert\, (X_1,\ldots,X_n)\in\mathcal{S}^n(X)\right\},
\]
is called the \emph{inf-convolution} of $\{\rho_1,\dots,\rho_n\}$. For $n=2$ we simply write $\rho_1\ic\rho_2$.
\end{definition}

\smallskip

\begin{remark}\label{rm:ic sum}
Recall that, if $\rho_1,\dots,\rho_n$ are monetary risk measures, then for every $X\in L^1$
\[
\icn{i}{n}\rho_i(X)=\inf\{m\in\R \mid X-m\in\cA_1+\cdots+\cA_n\}
\]
where $\cA_i=\{X\in L^1 \mid \rho_i(X)\leq0\}$ is the acceptance sets induced by $\rho_i$ for $i=1,\dots,n$. This shows that the infimal convolution of monetary risk measures is also a monetary risk measure.
\end{remark}

\smallskip

We establish a general inequality for inf-convolutions. More precisely, we show that any inf-convolution of adjusted ES's can be controlled from below by a suitable adjusted ES. This allows us to derive a formula for the inf-convolution of an adjusted ES with itself.

\begin{proposition}\label{prop:selfconv}
Let $n\in\N$ and consider the risk profiles $g,g_1,\dots,g_n\in\cG$. For every $X\in L^1$
\begin{equation}
\label{eq: inf conv 1}
\icn{i}{n}\ES^{g_i}(X) \geq \ES^{\sum_{i=1}^ng_i}(X).    
\end{equation}
In particular, for every $X\in L^1$
\begin{equation}
\label{eq: inf conv 2}
\icn{i}{n}\rg(X) = \ES^{ng}(X).
\end{equation}
\end{proposition}
\begin{proof}
To show \eqref{eq: inf conv 1}, it suffices to focus on the case $n=2$. For all $Y\in L^1$ and $p\in[0,1]$ we have
\[
\ES^{g_1}(Y)+\ES^{g_2}(X-Y) \geq \ES_p(Y)-g_1(p)+\ES_p(X-Y)-g_2(p) \geq \ES_p(X)-(g_1+g_2)(p)
\]
by subadditivity of ES. Taking the supremum over $p$ and the infimum over $Y$ delivers the desired inequality. To show \eqref{eq: inf conv 2}, note that the inequality ``$\geq$'' follows directly from \eqref{eq: inf conv 1}. To show the inequality ``$\leq$'', observe that
\[
\rg\bigg(\frac{1}{n}X\bigg) = \frac{1}{n}\sup_{p\in[0,1]}\{\ES_p(X)-ng(p)\} = \frac{1}{n}\ES^{ng}(X).
\]
As a result, we infer that
\[
\icn{i}{n}\rg(X) \leq \sum_{i=1}^n\rg\bigg(\frac{1}{n}X\bigg) = \ES^{ng}(X).
\]
This yields the desired inequality and concludes the proof.
\end{proof}

\smallskip

\begin{remark}
A risk measure that is not subadditive may incentivize the splitting and (internal) reallocation of risk with the sole purpose of reaching a lower level of capital requirements. This is related to the notion of regulatory arbitrage introduced in \cite{Wang_reg}. In line with that paper, we say that a functional $\rho:L^1\to(-\infty,\infty]$ is either {\em free} of regulatory arbitrage or has {\em limited} or {\em infinite} regulatory arbitrage if the quantity (recall our convention $\infty-\infty=-\infty$)
\[
\rho(X)-\inf_{n\in\N}\icn{i}{n}\rho(X)
\]
is null, finite, or infinite for every $X\in L^1$. Clearly, every risk measure that is not subadditive admits regulatory arbitrage. The preceding result on infimal convolutions allows us to show that an adjusted ES exhibits regulatory arbitrage only in a limited form. More precisely, for a risk profile $g\in\cG$ with $g(0)=0$ we have:
\begin{enumerate}[(i)]
    \item $\rg(X)-\inf_{n\in\N}\icn{i}{n}\rg(X)<\infty$ for every $X\in L^1$ with $\rg(X)<\infty$.
    \item $\rg(X)-\inf_{n\in\N}\icn{i}{n}\rg(X)=\infty$ for every $X\in L^1$ with $\rg(X)=\infty$.
\end{enumerate}
In particular, to prove (i), it suffices to note that Proposition~\ref{prop:selfconv} implies for every $X\in L^1$
\[
\inf_{n\in\N}\icn{i}{n}\rg(X) = \inf_{n\in\N}\ES^{ng}(X) \geq \ES_0(X)=\E[X] > -\infty.
\]
\end{remark}

\smallskip

We conclude this section by focusing on dual representations, which are a useful tool in many applications, notably optimization problems; see the general discussion in \cite{Rockafellar1974} and the results on risk measures in \cite{FS16}. In what follows we denote by $\cP$ the set of probability measures on $(\Omega,\cF)$ and use standard notation for Radon-Nikodym derivatives.

\begin{proposition}
\label{prop:simpledual}
Consider a risk profile $g\in\cG$. For every $X\in L^1$ we have
\[
\rg(X) = \sup_{\probq\in\cP^\infty_\probp}\bigg\{\E_\probq[X]-
g\bigg(1-\bigg\|\frac{d\probq}{d\probp}\bigg\|^{-1}_\infty\bigg)\bigg\},
\]
where $\cP^\infty_\probp=\{\probq\in\cP \mid \probq\ll\probp, \ d\probq/d\probp\in L^\infty\}$.
\end{proposition}
\begin{proof}
For notational convenience, for every $\probq\in\cP^\infty_\probp$ set
\[
D(\probq) = \bigg\{p\in[0,1] \,\bigg\vert\,	\frac{d\probq}{d\probp}\leq\frac{1}{1-p}\bigg\} = \bigg[1-\bigg\|\frac{d\probq}{d\probp}\bigg\|^{-1}_\infty,1\bigg].
\]
Take $X\in L^1$. The well-known dual representation of ES states that
\[
\ES_p(X) = \sup\bigg\{\E_\probq[X] \,\Big\vert\, \probq\in\cP^\infty_\probp,\  \tfrac{d\probq}{d\probp}\leq\frac{1}{1-p}\bigg\}
\]
for every $p\in[0,1]$; see, e.g., \cite{FS16}. Then, it follows that
\begin{eqnarray*}
\rg(X)
&=&
\sup_{p\in[0,1]}\bigg\{\sup_{\probq\in\cP^\infty_\probp, \ p\in D(\probq)}\{\E_\probq[X]-g(p)\}\bigg\} \\
&=&
\sup_{\probq\in\cP^\infty_\probp}\bigg\{\sup_{p\in D(\probq)}\{\E_\probq[X]-g(p)\}\bigg\} \\
&=&
\sup_{\probq\in\cP^\infty_\probp}\bigg\{\E_\probq[X]-\inf_{p\in D(\probq)}g(p)\bigg\}.
\end{eqnarray*}
It remains to observe that the above infimum equals $g(1-\|d\probq/d\probp\|^{-1}_\infty)$
by monotonicity of $g$.
\end{proof}


\section{Benchmark-adjusted ES}
\label{sect: SSD}

In this section we focus on a special class of adjusted ES's for which the target risk profiles are expressed in terms of the ES profile of a reference random loss. As shown below, these special adjusted ES's are intimately linked with second-order stochastic dominance.

\begin{definition}
Consider a functional $\rho:L^1\to(-\infty,\infty]$.
\begin{enumerate}[(1)]
\item $\rho$ is called a {\em benchmark-adjusted ES} if there exists $Z\in L^1$ such that for every $X\in L^1$
\[
\rho(X) = \sup_{p\in [0,1]}\{\ES_p(X)-\ES_p(Z)\}.
\]
\item $\rho$ is called an {\em SSD-based risk measure} if there exists $Z\in L^1$ such that for every $X\in L^1$
\[
\rho(X) = \inf\{m\in\R \mid X-m\SSD Z\}.
\]
\end{enumerate}
\end{definition}

\smallskip

It is clear that benchmark-adjusted ES's are special instances of adjusted ES's for which the target risk profile is defined in terms of the ES profile of a benchmark random loss. The distribution of this random loss may correspond, for example, to the (stressed) historical loss distribution of the underlying position or to a target (risk-class specific) loss distribution. It is also clear that SSD-based risk measures are nothing but monetary risk measures associated with acceptance sets defined through second-order stochastic dominance.

\smallskip

The classical characterization of second-order stochastic dominance in terms of ES can be used to show that benchmark-adjusted ES's coincide with SSD-based risk measures. In addition, we provide a simple characterization of this class of risk measures.

\begin{theorem}
\label{prop:max}
For a monetary risk measure $\rho:L^1\to(-\infty,\infty]$ the following are equivalent:
\begin{enumerate}[(i)]
  \item $\rho$ is a benchmark-adjusted ES.
  \item $\rho$ is an SSD-based risk measure.
   \item $\rho$ is consistent with $\SSD$ and the set $\{X\in L^1 \mid \rho(X)\leq0\}$ has an $\SSD$-minimum element.
\end{enumerate}
\end{theorem}
\begin{proof}
Recall that for all $X\in L^1$ and $Z\in L^1$ we have $X\SSD Z$ if and only if $\ES_p(X)\leq\ES_p(Z)$ for every $p\in[0,1]$; see, e.g., Theorem 4.A.3 in \cite{bookCX}. For convenience, set $\cA=\{X\in L^1 \mid \rho(X)\leq0\}$. To show that (i) implies (ii), assume that $\rho$ is a benchmark-adjusted ES with respect to $Z\in L^1$. Then, for every $X\in L^1$
\begin{align*}
\rho(X) &= 
\inf\{m\in\R \mid X-m\in\cA\} \\
&= \inf\{m\in\R \mid \ES_p(X)-m\leq\ES_p(Z), \ \forall p\in[0,1]\} \\
&= \inf\{m\in\R \mid X-m\SSD Z\}.
\end{align*}
To show that (ii) implies (i), assume that $\rho$ is SSD-based with respect to $Z\in L^1$. Then, we have
\begin{align*}
\rho(X) &= 
\inf\{m\in\R \mid X-m\SSD Z\} \\
&= \inf\{m\in\R \mid \ES_p(X)-m\leq\ES_p(Z), \ \forall p\in[0,1]\} \\
&= \sup_{p\in[0,1]}\{\ES_p(X)-\ES_p(Z)\}.
\end{align*}
It is clear that (iii) implies (ii). Finally, to show that (ii) implies (iii), assume that $\rho$ is an SSD-based risk measure with respect to $Z\in L^1$. It is clear that $Z\in\cA$. Now, take an arbitrary $X\in\cA$. We find a sequence $(m_n)\subset\R$ such that $m_n\downarrow\rho(X)$ and $X-m_n\SSD Z$ for every $n\in\N$. This implies that $X-\rho(X)\SSD Z$. Since $\rho(X)\leq0$, we infer that $X\SSD Z$ as well. This shows that $\cA$ has an SSD-minimum element. To establish that $\rho$ is consistent with $\SSD$, take arbitrary $X,Y\in L^1$ satisfying $X\SSD Y$. For every $m\in\R$ such that $Y-m\SSD Z$ we clearly have that $X-m\SSD Y-m\SSD Z$. This implies that $\rho(X)\leq\rho(Y)$ and concludes the proof.
\end{proof}

\smallskip

The preceding result delivers an interesting representation of a benchmark-adjusted ES in terms of utility functions which helps highlighting its ``risk aversion'' nature. More precisely, we show that an adjusted ES with risk profile given by the ES profile of a benchmark random loss $Z\in L^1$ determines the minimal amount of capital that makes {\em every} risk-averse agent better off than being exposed to the loss $Z$. In this sense, one may view a benchmark-adjusted ES as a worst-case utility-based risk measure over all conceivable risk-averse profiles. Recall that, if one moves from utility functions to loss functions, then utility-based risk measures correspond to the so-called shortfall risk measures as defined, e.g., in \cite[Section 4.9]{FS16}.

\begin{proposition}
Let $Z\in L^1$ and consider the risk profile $g(p)=\ES_p(Z)$ for every $p\in[0,1]$. Moreover, let $\cU$ be the family of all (nonconstant) concave and increasing functions $u:\R\to\R$. Then, for every $X\in L^1$
\[
\ES^g(X) = \sup_{u\in\cU}\inf\{m\in\R \mid \E[u(m-X)]\geq\E[u(-Z)]\}.
\]
\end{proposition}
\begin{proof}
Let $\cA=\{X\in L^1 \mid \forall p\in[0,1], \ \ES_p(X)\leq\ES_p(Z)\}$ and set $\cA_u=\{X\in L^1 \mid \E[u(-X)]\geq\E[u(-Z)]\}$ for every $u\in\cU$. To establish the claim, we can equivalently prove that for every $X\in L^1$
\begin{equation}
\label{eq: sup of shortfall risk measures}
\inf\{m\in\R \mid X-m\in\cA\} = \sup_{u\in\cU}\inf\{m\in\R \mid X-m\in\cA_u\}.
\end{equation}
To this effect, Theorem 4.A.3 in \cite{bookCX} implies that
\[
\cA = \{X\in L^1 \mid X\SSD Z\} = \{X\in L^1 \mid \forall u\in\cU, \ \E[u(-X)]\geq\E[u(-Z)]\} = \bigcap_{u\in\cU}\cA_u.
\]
This implies \eqref{eq: sup of shortfall risk measures}. Indeed, the inequality ``$\geq$'' is clear. To show the inequality ``$\leq$'', take any number $k>\sup_{u\in\cU}\inf\{m\in\R \mid X-m\in\cA_u\}$. Then, for every $u\in\cU$ we must have $X-k\in\cA_u$ or, equivalently, $X-k\in\cA$. This yields $k\geq\inf\{m\in\R \mid X-m\in\cA\}$. Taking the infimum over such $k$'s delivers the desired inequality and completes the proof.
\end{proof}

\smallskip

In light of the relevance of benchmark adjusted ES's, we are interested in characterizing when the acceptable risk profile $g$ of an adjusted ES can be expressed in terms of an ES profile. To this effect, it is convenient to introduce the following additional class of risk measures, which will be shown to contain all benchmark-adjusted ES's. We denote by $L^0$ the space of all random variables.

\begin{definition}
A functional $\rho:L^1\to(-\infty,\infty]$ is called a {\em quantile-adjusted ES} if there exists $Z\in L^0$ such that for every $X\in L^1$
\[
\rho(X) = \sup_{p\in [0,1]}\{\ES_p(X)-\VaR_p(Z)\}.
\]
\end{definition}

\smallskip

To establish our desired characterization, for a risk profile $g\in\cG$ we 
define $h_g:[0,1]\to(-\infty,\infty]$ by
$$
h_g(p) := (1-p)g(p).
$$
Here, we set $0\cdot\infty=0$ so that $h_g(1)=0$. Moreover, we introduce the following sets:
$$
\mathcal{G}_{\VaR} := \{g\in\mathcal G \mid \mbox{$g$ is finite on $[0,1)$, left-continuous on $[0,1]$, and right-continuous at $0$}\},
$$
$$
\mathcal{G}_{\ES} := \{g\in\mathcal G_{\VaR} \mid \mbox{$h_g$ is concave on $(0,1)$ and left-continuous at $1$}\}.
$$

\smallskip

\begin{lemma}
\label{lem: characterizing profiles}
For every risk profile $g\in\cG$ the following statements hold:
\begin{enumerate}[(i)]
\item $g\in \mathcal G_{\VaR}$ if and only if there exists a random variable $Z\in L^0$ that is bounded from below and satisfies $g(p)=\VaR_p(Z)$ for every $p\in [0,1]$.
\item $g\in \mathcal G_{\ES}$ if and only if there exists a random variable $Z\in L^1$ such that $g(p)={\ES}_p(Z)$ for every $p\in [0,1]$.
\end{enumerate}
\end{lemma}
\begin{proof}
{\em (i)} The ``if'' part is clear. For the ``only if'' part, let $U$ be a uniform random variable on $[0,1]$ and set $Z=g(U)$. Then, it is well known that $\VaR_p(Z)=g(p)$ for every $p\in[0,1]$. Moreover, since $g(0)>-\infty$, we see that $Z$ is bounded from below.

\smallskip

{\em (ii)} The ``if'' part is straightforward. For the ``only if'' part, let $U$ be a uniform random variable on $[0,1]$. We denote by $h_g'$ the left derivative of $h_g$. Then, for every $p\in[0,1)$ we have
$$
\ES_p(-h_g'(U)) = -\frac{1}{1-p}\int_p^1 h_g'(u)\d u = -\frac{h_g(1)-h_g(p)}{1-p} = g(p).
$$
This shows that, by taking $Z=-h_g'(U)$, we have $g(p)=\ES_p(Z)$ for every $p\in [0,1)$. The left continuity of $g$ and $\ES_{\cdot}(Z)$ at $1$ gives the same equality for $p=1$.
\end{proof}

\smallskip

As a direct consequence of the previous lemma we derive a characterization of quantile- and benchmark-adjusted ES's in terms of the underlying risk profile.

\begin{theorem}
\label{prop: on VaR- and ES-adjusted ES}
For every risk profile $g\in\cG$ the following statements hold:
\begin{enumerate}[(i)]
\item There exists $Z\in L^0$ that is bounded from below and such that $\rg$ is a quantile-adjusted ES with respect to $Z$ if and only if $g\in\mathcal{G}_{\VaR}$.
\item There exists $Z\in L^1$ such that $\rg$ is an benchmark-adjusted ES with respect to $Z$ if and only if $g\in\mathcal{G}_{\ES}$.
\end{enumerate}
\end{theorem}

\smallskip

\begin{remark}
We infer from Theorem \ref{prop:max} and \ref{prop: on VaR- and ES-adjusted ES} that the classical ES does not belong to the class of SSD-based risk measures as the associated risk profile is not in $\mathcal G_{\ES}$ (see also Proposition \ref{prop:coherent}).
\end{remark}

\smallskip

Since we clearly have $\cG_{\ES}\subset\cG_{\VaR}$, it follows from the above results that every benchmark-adjusted ES is also a quantile-adjusted ES. In particular, this implies that, for every random variable $Z\in L^1$, we can always find a random variable $W\in L^0$ such that $\VaR_p(W)=\ES_p(Z)$ for every $p\in[0,1]$. In words, every ES profile can be reproduced by a suitable VaR profile. As pointed out by the next proposition, the converse result is, in general, not true. In addition, we also show that an adjusted ES need not be a quantile-adjusted ES.

\begin{proposition}
\begin{enumerate}[(i)]
    \item There exists $g\in\cG$ such that $\rg\neq\rh$ for every $h\in\cG_{\VaR}$.
    \item There exists $g\in\cG_{\VaR}$ such that $\rg\neq\rh$ for every $h\in\cG_{\ES}$.
\end{enumerate}
\end{proposition}
\begin{proof}
The second assertion follows immediately from Theorem \ref{prop: on VaR- and ES-adjusted ES} and the fact that the inclusion $\cG_{\ES}\subset\cG_{\VaR}$ is strict. To establish the first assertion, fix $q\in(0,1)$ and define $g\in\cG$ by setting
\[
g(p)=
\begin{cases}
0 & \mbox{if} \ p\in[0,q],\\
\infty & \mbox{if} \ p\in(q,1].
\end{cases}
\]
It follows that
$$
\rg(X) = \sup_{p\in [0,q]}\{\ES_p(X)\} = \ES_q(X)
$$
for every $X\in L^1$. We claim that $\rg$ is not a quantile-adjusted ES. To the contrary, suppose that there exists a random variable $Z\in L^0$ that is bounded from below and satisfies
$$
\ES_q(X) = \rg(X) = \sup_{p\in [0,1]}\{\ES_p(X)-\VaR_p(Z)\}
$$
for every $X\in L^1$. Take $r\in(q,1)$ and $X\in\X$ such that $\ES_r(X)>\ES_q(X)$. Then, for each $\lambda>0$
\begin{align*}
\ES_q(X) &= \frac 1\lambda \ES_q(\lambda X) = \frac 1\lambda \sup_{p\in [0,1]}\{\ES_p(\lambda X)-\VaR_p(Z)\} \\
&\ge \frac 1\lambda (\ES_r(\lambda X)-\VaR_r(Z)) = \ES_r(X)-\frac{1}{\lambda}\VaR_r(Z).
\end{align*}
By sending $\lambda\to\infty$, we obtain $\ES_q(X)\ge\ES_r(X)$, which contradicts our assumption on $X$.
\end{proof}

\smallskip

Note that ES is always finite on our domain. Here, we are interested in discussing the finiteness of adjusted ES's associated with risk profiles in the class $\cG_{\VaR}$ and $\cG_{\ES}$. We show that finiteness on the whole reference space $L^1$ can never hold in the presence of a risk profile in $\cG_{\ES}$ while it can hold if we take a risk profile in $\cG_{\VaR}$.

\begin{proposition}\label{prop:3}
Consider a risk profile $g\in\cG$. If $g\in\cG_{\VaR}$, then $\rg$ can be finite on $L^1$. If $g\in\cG_{\ES}$, then $\rg$ cannot be finite on $L^1$.
\end{proposition}
\begin{proof}
To show the first part of the assertion, set $g(p)=\frac{1}{1-p}$ for every $p\in[0,1]$ (with the convention $\frac{1}{0}=\infty$). Note that $g\in\cG_{\VaR}$. Fix $X\in L^1$ and note that there exists $q\in (0,1)$ such that
\[
\sup_{p\in[q,1]}\int_{p}^1\VaR_r(X)\d r<1.
\]
It follows that
$$
\sup_{p\in [q,1]}\left\{\ES_p(X)-\frac{1}{1-p}\right\} = \sup_{p\in [q,1]}\left\{\frac1 {1-p}\left(\int_p^1 \VaR_r(X)\d r - 1\right)\right\} \le 0 .
$$
Therefore,
$$
\rg(X) \le \max\left\{\sup_{p\in [0,q]}\left\{{\ES}_p(X)-\frac{1}{1-p}\right\},0\right\} \le \max\{{\ES}_q(X),0\} <\infty.
$$
This shows that $\rg$ is finite on the entire $L^1$. To establish the second part of the assertion, take $Z\in L^1$ and set $g(p)=\ES_p(Z)$ for every $p\in[0,1]$. Note that $g\in\cG_{\ES}$ by Lemma~\ref{lem: characterizing profiles}. If $Z$ is bounded from above, then take $X\in L^1$ that is unbounded from above. In this case, it follows that
\[
\rg(X) \ge \ES_1(X)-\ES_1(Z) = \infty.
\]
If $Z$ is unbounded from above, then take $X=2Z\in L^1$. In this case, we have
\[
\rg(X) \ge \ES_1(2Z)-\ES_1(Z) = \ES_1(Z) = \infty.
\]
Hence, we see that $\rg$ is never finite on $L^1$.
\end{proof}

\smallskip

The next result improves Proposition~\ref{prop:selfconv} by showing that the inf-convolution of benchmark-adjusted ES's can still be expressed as an adjusted ES.

\begin{proposition}
Let $n\in\N$ and consider the risk profiles $g_1,\dots,g_n\in\cG_{\ES}$. For every $X\in L^1$
\[
\icn{i}{n}\ES^{g_i}(X) = \ES^{\sum_{i=1}^ng_i}(X).
\]
\end{proposition}
\begin{proof}
The inequality ``$\geq$'' follows from Proposition~\ref{prop:selfconv}. To show the inequality ``$\leq$'', note that there exist $Z_1,\dots,Z_n\in L^1$ such that $\cA_{g_i}=\{X\in L^1 \mid X\SSD Z_i\}$ by Theorem \ref{prop: on VaR- and ES-adjusted ES}. We prove that
\[
\cA:=\{X\in L^1 \mid \rng{\sum_{i=1}^{n}g_i}(X)\le 0\}\subset \sum_{i=1}^n\cA_{g_i}
\]
which, together with Remark \ref{rm:ic sum}, yields the desired inequality. Let $U$ be a uniform random variable and, for any $X\in L^1$, denote by $F^{-1}_X$ the (left) quantile function of $X$. Take $i\in\{1,\ldots,n\}$ and note that $F^{-1}_{Z_i}(U)\sim Z_i$. It follows from the law invariance of $\ES$ that $\ES_p(F^{-1}_{Z_i}(U))=\ES_p(Z_i)$ for every $p\in[0,1]$, so that $F^{-1}_{Z_i}(U)\in\cA_{g_i}$. Since the random variables $F^{-1}_{Z_i}(U)$'s are comonotonic,
\[
\sum_{i=1}^n \ES_p(Z_i)=\sum_{i=1}^n \ES_p(F^{-1}_{Z_i}(U))=\ES_p(Z)
\]
with $Z=\sum_{i=1}^n F^{-1}_{Z_i}(U)$.
We deduce that each $X\in\cA$ satisfies $\ES_{p}(X)\le\ES_p(Z)$ for every $p\in[0,1]$, which is equivalent to $X\SSD Z$. Note that $Z\in \sum_{i=1}^n\cA_{g_i}$ so that $\icn{i}{n}\ES^{g_i}(Z)\le 0$. Since the inf-convolution is consistent with $\SSD$, as shown in Theorem 4.1 by \cite{MW16}, we have  $\icn{i}{n}\ES^{g_i}(X)\le\icn{i}{n}\ES^{g_i}(Z)\le 0$, which implies $X\in\sum_{i=1}^n\cA_{g_i}$ as desired.
\end{proof}

\section{Optimization with benchmark-adjusted ES}
\label{sect: SSD optimization}

Using the characterization of benchmark-adjusted ES's established in Theorem \ref{prop:max}, many optimization problems related to benchmark-adjusted ES's or, equivalently, SSD-based risk measures can be solved explicitly. In this section, we focus on risk minimization and utility maximization problems in the context of a multi-period frictionless market that is complete and arbitrage free. The interest rate is set to be zero for simplicity. As is commonly done in the literature, this type of optimization problems, which are naturally expressed in terms of dynamic investment strategies, can be converted into static optimization problems by way of martingale methods. Below we focus directly on their static counterparts. For more details we refer, e.g., to \cite{SFW09} or \cite{FS16}. In addition, to ensure that all our problems are well defined, we work in the space $L^\infty$ of $\mathbb{P}$-bounded random variables.

\smallskip

In the sequel, we denote by $\mathbb Q$ the risk-neutral pricing measure (whose existence and uniqueness in our setting are ensured by the Fundamental Theorem of Asset Pricing), by $w\in \R$ a fixed level of initial wealth, by $x\in \R$ a real number representing a constraint, by $u:\R\to \R\cup\{-\infty\}$ a concave and increasing function that is continuous (at the point where it potentially jumps to $-\infty$) and satisfies $\lim_{y\to-\infty}u(y)<x<\lim_{y\to\infty}u(y)$, and by $\rho:L^\infty\to(-\infty,\infty]$ a risk functional. We focus on the following five optimization problems:
\begin{enumerate}[(A)]
\item Risk minimization with a budget constraint:
$$
\mbox{minimize } \rho(X) \mbox{~over $X\in L^\infty$~subject to $\E_{\Q}[w-X]\le x$.}
$$\item Price minimization with   controlled risk:
$$
\mbox{minimize } \E_{\Q}[w-X]   \mbox{~over $X\in L^\infty$~subject to $\rho(X)\le x$.}
$$
\item Risk minimization with a target utility level:
$$
\mbox{minimize } \rho(X) \mbox{~over $X\in L^\infty$~subject to $\E[u(w-X)]= x$}.
$$
\item Worst-case utility with a reference risk assessment:
$$
\mbox{minimize } \E[u(w-X)] \mbox{~over $X\in L^\infty$~subject to $\rho(X)= x$.}
$$
\item Worst-case risk  with a reference risk assessment:
$$
\mbox{maximize } \rho'(X) \mbox{~over $X\in L^\infty$~subject to $\rho(X)= x$,}
$$
where $\rho'$ is an SSD-consistent functional that is continuous with respect to the $L^\infty$-norm.
\end{enumerate}

\smallskip

Problem (A) is an optimal investment problem minimizing the risk given a budget constraint. Conversely, problem (B) aims at minimizing the cost given a controlled risk level. Problem (C) is about minimizing the risk exposure with a target utility level, similar to the mean-variance problem of \cite{M52}. The interpretation of problems (D) and (E) is different from the first three problems: They are not about optimization over risk, but about ambiguity, i.e., in these problems the main concern is model risk. Indeed, the set $L^\infty$ may represent the class of plausible models for the distribution of a certain financial position of interest. In the case of problem (D), the assumption is that the only available information for $X$ is the risk figure $\rho(X)$, evaluated, e.g., by an expert or another decision maker. In this context, we are interested in determining the worst case utility among all possible models which agree with the evaluation $\rho(X)=x$ (see also Example 5.3 of \cite{WXZ19}). A similar interpretation can be given for problem (E).

\begin{proposition}\label{prop:opt}
Each of the optimization problems (A)-(E) relative to a benchmark-adjusted ES $\rho=\ES^g$ for $g\in\cG_{\ES}$ admits an optimal solution of the explicit form $Z+z$ where $Z\in L^\infty$ has the ES profile $g$ and $z\in\R$. Moreover, $Z$ is comonotonic with $\frac{d\Q}{d\p}$ in (A)-(B), and the (binding) constraint uniquely determines $z$ in each problem.
\end{proposition}
\begin{proof}
The result for the optimization problem (A) is a direct consequence of Proposition 5.2 in \cite{MW16}. Let $Z$ be comonotonic with $d\Q/d\p$ which has ES profile $g$ (comonotonicity is only relevant in problems (A) and (B)). Note that $\rho(Z)=0$. For any random variable $X\in L^\infty$, we set $Y_X=Z+\rho(X)$. It is clear that $\rho(Y_X) = \rho(X)$ and
$$
\ES_p(Y_X) =g(p) + \rho(X) = g(p) + \sup_{q\in [0,1]} \{\ES_q(X) -g(q)\} \ge \ES_p(X).
$$
Hence, $X\SSD Y_X$. This observation will be useful in the analysis below.
\begin{enumerate}[(i)]
\item
We first look at problem (B).
First, since both $X\mapsto \E_\Q[X]$ and $\rho$ are translation-invariant, the condition $\rho(X)\le x$ is binding, and problem (B) is equivalent to
maximizing $\E_\Q [X]$ over $X\in L^\infty$ such that $\rho(X)= x$.
Let $X\in L^\infty$ be any random variable with $\rho(X)=x$ and let $X'$ be identically distributed as $X$ and comonotonic with $d\Q/d\p$. Since $X'\sim X$, by the Hardy--Littlewood inequality (see, e.g., Remark 3.25 of \cite{R13}),
we have $\E_\Q[X] \le \E_\Q[X']$.
Moreover, for any random variable $Y\in L^\infty$ that is comonotonic with $d \Q/d \p$, we can write (see, e.g., (A.8) of \cite{MW16})
$$
\E_\Q[Y] =
\int_0^1 \ES_p(Y) \d \mu (p)
$$
for some Borel probability measure $\mu$ on $[0,1]$. Hence, $X'\SSD Y_X$ implies $\E_\Q[X'] \le \E_\Q[Y_X]$, and we obtain
$$
\E_\Q[X] \le \E_\Q[X'] \le \E_\Q[Y_X].
$$
Note also that $\rho(Y_X)=\rho(X)=x$. Hence, for any random variable $X\in L^\infty$, there exists $Z+z$ for some $z\in \R$ which dominates $X$ for problem (B). Since both the constraint and the objective are continuous in $z\in \R$, an optimizer of the form $Z+z$ exists.
\item
We next look at problem (C).
Let $X\in L^\infty$ be any random variable such that $\E[u(w-X)] =x$.
The aforementioned fact $X\SSD Y_X$ implies that
$\E[u(w-Y)] \le \E[u(w-X)]=x$ since $u$ is a concave utility function.
Therefore, there exists $\varepsilon\geq0$ such that $\E[u(w-(Y-\varepsilon))]=x$, and we take the largest $\varepsilon$ satisfying this equality, which is obviously finite. Let $z=\rho(X)-\varepsilon$. It is then clear that $\E[u(w-(Z+z))]=\E[u(w-X)]=x$ and $\rho(Z+z)= \rho(Y-\varepsilon) =\rho(X)-\varepsilon \le \rho(X)$. Hence, $Z+z$ dominates $X$ as an optimizer for problem (C). Since both the constraint and the objective are continuous in $z\in \R$, an optimizer of the form $Z+z$ exists.
\item Problems (D) and (E) can be dealt with using similar arguments.\qedhere
\end{enumerate}
\end{proof}

\smallskip

\begin{remark}
(i) Recall that ES does not belong to the class of SSD-based risk measures. As a consequence, the results in this section do not directly apply to ES. In particular, although ES is consistent with SSD, its acceptance set does not have a minimum SSD element as required by Proposition \ref{prop:max}. We refer to \cite{WZ20} for a different characterization of ES.

\smallskip

(ii) In the context of decision theory and, specifically, portfolio selection, it is sometimes argued that (second order) stochastic dominance is too extreme in the sense that it ranks risks according to the simultaneous preferences of \emph{every} risk-averse agent, thus including utility functions that may lead to counterintuitive outcomes. A typical example is the one proposed by \cite{LL02}. Consider a portfolio that pays one million dollars in $99\%$ of cases and nothing otherwise and another portfolio that pays one dollar with certainty. According to the sign convention adopted in this paper, the corresponding payoffs are given by
\[
X=\begin{cases}
0 & \mbox{with probability $1\%$}\\
-10^6 & \mbox{with probability $99\%$}
\end{cases}
 \ \ \ \mbox{and} \ \ \ Y=-1.
\]
Even though $X$ does not dominate $Y$ with respect to SSD, most agents prefer $X$ to $Y$. Thus, the authors argue for the necessity of relaxing SSD in favor of a more reasonable notion. We point out that our approach yields a novel and reasonable generalization of SSD. First, consider the risk profile defined by $g(p)=\ES_p(Y)=-1$ for every $p\in[0,1]$ and note that $X$ is acceptable under $\ES^g$ precisely when $X\SSD Y$. Note also that
\[
\ES_p(X)\le g(p) \ \iff \ p\le\bar{p}:=1-\frac{10^{-4}}{10^6-1}\approx1-10^{-10}.
\]
This fact has two implications. On the one hand, it confirms that $X$ does not dominate $Y$ with respect to SSD and highlights that this failure is due to the behavior of $X$ in the far region of its left tail. On the other hand, it suggests that it is enough to consider the new risk profile defined by $h(p)=g(p)$ for $p\le\bar{p}$ and $h(p)=\infty$ otherwise to make $X$ acceptable under $\ES^h$. In other words, moving from $g$ to $h$ is equivalent to moving from SSD to a relaxed form of SSD that enlarges the spectrum of acceptability in portfolio selection problems. However, note that $\ES^h$ is not an SSD-based risk measure and, hence, the existence results obtained above do not apply to it. A systematic study of optimization problems under constraints of $\ES^h$ type requires further research. 
\end{remark}

{\small
\bibliographystyle{apalike}
\bibliography{biblioRM}
}

\end{document}